\documentclass[journal]{IEEEtran}
\usepackage{amsmath,amssymb,amsfonts,amsthm}
\usepackage{booktabs}
\usepackage{graphicx}
\usepackage{algorithm}
\usepackage{algorithmic}
\usepackage{tikz}
\usetikzlibrary{arrows.meta,positioning}
\usepackage{hyperref}

\newcommand{\rL}{\rho(\mathbf{L})}
\newtheorem{proposition}{Proposition}
\newtheorem{remark}{Remark}

\begin{document}

\title{The Loop-Gain Matrix: Coupled Rebalancing Feedback and the
Blind Spots of Scalar Stability Monitoring}

\author{Jihwan~Woo%
\thanks{J.~Woo is a Senior Specialist Solutions Architect, AI/ML, with
Amazon Web Services. The views expressed are those of the author and do
not represent those of Amazon Web Services or its affiliates.}%
\thanks{Working paper, \today. Code, seeds, and result files for every
reported number are archived with the project repository and available
from the author.}}

\maketitle

\begin{abstract}
The stability of markets hosting leveraged exchange-traded products is
governed not by any single product's loop gain but by the spectral radius
of a loop-gain \emph{matrix}, and scalar per-product monitoring
underestimates system feedback by construction. Recent work measures the
self-reinforcement of a leveraged fund's daily close rebalancing through a
scalar loop gain $\ell = \Lambda_c K$ (closing-venue price impact times
rebalancing capital) and treats cross-asset spillovers as estimation bias.
We model complexes of leveraged products written on correlated underlyings
as a coupled feedback system with matrix gain $\mathbf{L}$, and show that
scalar monitoring has two distinct blind spots: (i) \emph{cycle
amplification}, because $\rL \ge \max_i \ell_{ii}$ for nonnegative
coupling, with strict excess under two-way coupling over the coupled
pair's own gains; and (ii)
\emph{transmitted displacement}, which arises under one-way coupling and
is invisible to the receiving asset's own gain. We give a reduced-form
estimator of $\mathbf{L}$ that requires only prices and public fund
assets---no signed order flow---via cross-asset overnight reversals,
reporting its measurement-convention sensitivity explicitly, and
validate it in simulation: the spectral radius is recovered with RMSE
0.005 at $T{=}250$, a lead--lag confounder produces a 2\% false-alarm
rate, and in a calibrated blind-spot configuration the scalar monitor
reports ``safe'' on 100\% of paths while the matrix monitor reports
``unsafe'' on 100\%. Empirically, using the 2026 Korean single-stock
LETF episode, we detect the transmission channel from the SK~Hynix complex
into Samsung Electronics' closing price (difference-in-differences
$z=-2.82$; exact randomization $p=0.0055$ against 182 control pairs;
block-bootstrap 95\% CI excludes zero), scaling with the sender's observed
rebalancing capital. Under a conservative calibration, roughly 41\% of
Samsung's closing displacement variance is imported from the neighboring
complex---displacement that the stock's own ``moderate'' scalar gain
(0.24) cannot register. The same estimator returns nulls across the U.S.
MSTR--Bitcoin--Coinbase complex, whose rebalancing capital is comparable
but whose closing venue is far deeper, consistent with a graded response
surface. Stability monitoring of leveraged-product ecosystems should be
organized around the (complex $\times$ venue) matrix, not around products.
\end{abstract}


\section{Introduction}\label{sec:intro}

\IEEEPARstart{A}{leveraged} exchange-traded fund (LETF) that promises $L$
times the daily return of an underlying must trade $A(L^2-L)r$ near the
close of each day, where $A$ is fund assets and $r$ the day's return: a
mechanical, publicly predictable order that trades in the direction of the
day's move~\cite{cheng2009}. Because the coefficient $L^2-L$ is positive
for both leveraged ($L{=}2$) and inverse ($L{=}-2$) products, long and
short funds \emph{add} rather than cancel. When the aggregate order is
large relative to the venue that clears it, the closing price it
references is displaced by the order itself, creating a within-auction
feedback loop. Zhao~\cite{zhao2026} formalizes this self-reference
through a scalar \emph{loop gain} $\ell=\Lambda_c K$---the closing-venue
price-impact coefficient times the complex's total rebalancing capital
$K=\sum_f A_f(L_f^2-L_f)$---and shows that the extreme 2026 volatility of
the Korean market, which hosted the largest single-stock LETF complexes
on record, traces to loop gains an order of magnitude beyond U.S.\
levels. This unifies an older empirical dispute: studies reporting
late-day amplification~\cite{shum2016,tuzun2013} and studies reporting
benign, arbitraged-away impact~\cite{ivanov2018} occupy different regions
of one response surface indexed by $\ell$; a recent survey diagnoses
methodological weaknesses across this literature~\cite{lenkey2024}.

The scalar formulation, however, contains a structural assumption:
one risky asset. Real leveraged complexes are written on
\emph{correlated} underlyings that share clearing mechanisms, index
products, and arbitrage capital. In Korea, the two treated stocks
(Samsung Electronics and SK~Hynix) jointly constitute roughly half of
index capitalization, and their complexes rebalance in the same
ten-minute closing auction. In the U.S., the MicroStrategy, spot-Bitcoin,
and Coinbase complexes are levered claims on a common risk factor. When
flow driven by one asset's return displaces a correlated neighbor's
closing price---through index arbitrage, sector market-making inventory,
or common liquidity provision, the mechanisms documented in the
cross-impact literature~\cite{capponi2021,hasbrouck2001}---the feedback
is a \emph{coupled} dynamical system. Its gain is a matrix, and for
nonnegative coupling---products whose flows push neighbors the same
way, the empirically relevant case for common-factor
complexes---elementary spectral theory implies that per-product scalar
monitoring is guaranteed to understate system feedback. The observation is familiar in
engineering, where multivariable systems are never certified loop by
loop~\cite{astrom2008}; it has not, to our knowledge, been applied to
the leveraged-fund stability debate.

This paper makes the observation operational. Our contributions are
fourfold.

\emph{1) Framework} (Section~\ref{sec:model}). We write the closing
fixed point of $n$ coupled complexes as $\mathbf{r}_2 =
(\mathbf{I}-\mathbf{L})^{-1}(\mathbf{L}\mathbf{r}_1 + \mathbf{v})$ with
loop-gain matrix
$\mathbf{L}=\boldsymbol{\Phi}\,\mathrm{diag}(\boldsymbol{\gamma})$, where
$\boldsymbol{\Phi}$ collects own- and cross-impact coefficients and
$\gamma_j$ is complex $j$'s rebalancing capital scaled by liquidity.
Stability is governed by $\rL$, and for elementwise nonnegative
$\mathbf{L}$, $\rL \ge \max_i \ell_{ii}$~\cite{horn2013}. We distinguish
two blind spots of scalar monitoring: \emph{cycle amplification}
($\rL>\max_i\ell_{ii}$, requiring two-way coupling and strict for the
coupled pair at $n{=}2$;
Proposition~\ref{prop:rho}) and \emph{transmitted displacement} (nonzero
off-diagonal entries of $(\mathbf{I}-\mathbf{L})^{-1}-\mathbf{I}$,
arising already under one-way coupling; Proposition~\ref{prop:transmit}).

\emph{2) Identification without order-flow data}
(Section~\ref{sec:ident}). The displacement component of the close
reverts overnight. Regressing overnight returns on the vector of same-day
returns identifies
$\mathbf{M}\equiv(\mathbf{I}-\mathbf{L})^{-1}-\mathbf{I}$ up to the
overnight correction share, and hence $\mathbf{L} =
\mathbf{I}-(\mathbf{I}+\mathbf{M})^{-1}$
(Proposition~\ref{prop:ident}), using only prices and publicly disclosed
fund assets. This matters because the signed intraday order-flow data
used for direct impact estimation in~\cite{zhao2026} are available to few
researchers and to no real-time monitor; from a signal-processing
standpoint the problem is closed-loop system identification from output
data alone~\cite{ljung1999}, made tractable by the known, disclosed
structure of the feedback element.

\emph{3) Validation and falsification attempts}
(Section~\ref{sec:synth}). In a two-asset simulation with a closing fixed
point, the estimator recovers $\mathbf{L}$ and $\rL$ (RMSE 0.005 for
$\rL$ at $T{=}250$); a fundamental lead--lag confounder---the most
dangerous alternative, because cross-autocorrelation among correlated
equities is a documented regularity~\cite{lo1990,chordia2000}---yields a
2\% false-alarm rate, because continuation and reversal load with
opposite signs; and in a calibrated configuration with all own-gains
below threshold but $\rL$ above it, the scalar monitor is wrong on every
path and the matrix monitor right on every path
(Fig.~\ref{fig:blindspot}).

\emph{4) Evidence} (Sections~\ref{sec:korea}--\ref{sec:us}). Around the
Korean single-stock LETF launch (May~27,~2026), the closing price of
Samsung Electronics acquires a reversal loading on SK~Hynix's same-day
return that did not exist before launch, does not exist in placebo pairs,
scales with the Hynix complex's observed rebalancing capital, and appears
only in the direction predicted by relative complex size. Inference
survives exact randomization against all 182 ordered pairs of fourteen
non-treated large caps and a date-block bootstrap. A conservative
calibration attributes ${\sim}41\%$ of Samsung's closing displacement
variance to the imported channel. The same machinery returns nulls on the
U.S.\ MSTR--Bitcoin--Coinbase complex, where rebalancing capital is
comparable but the closing venue is an order of magnitude deeper:
coupling, like the scalar gain itself, is graded by venue impact.

\section{Related Work}\label{sec:related}

\subsection{LETF rebalancing and its market impact}

The debate over LETF rebalancing has unfolded in three acts. The
mechanics came first: \cite{cheng2009} derived the arithmetic
\eqref{eq:arith}, showed that leveraged and inverse products reinforce
rather than offset, and predicted end-of-day pressure proportional to the
day's move; \cite{tuzun2013} drew the analogy to the portfolio insurers
of October 1987, and \cite{shum2016} documented elevated late-day
volatility on high-rebalancing days. The rebuttal came from the
economics of predictable flow. If an order is perfectly foreseeable,
anticipatory liquidity should be waiting for it---the sunshine-trading
logic of~\cite{admati1991}, confirmed empirically for predictable
institutional demand by~\cite{bessembinder2016}---and \cite{ivanov2018}
indeed found that capital flows and anticipation largely neutralize
LETF pressure in U.S.\ data. Parallel measurements of predictable-flow
footprints---VIX products~\cite{brogger2021}, commodity and volatility
funds~\cite{todorov2024}, index rebalancing~\cite{petajisto2011,
pavlova2023}, and ETF ownership generally~\cite{bendavid2018}---returned
effects that are real but modest, and a survey of the LETF branch
concluded that many positive findings rest on fragile
designs~\cite{lenkey2024}. The third act is the reconciliation:
\cite{zhao2026} models arbitrageurs who prey on the mandated order,
shows both camps live on one response surface indexed by the scalar loop
gain $\ell=\Lambda_c K$, and measures the Korean 2026 episode at loop
gains far beyond anything in the U.S.\ record. Throughout all three
acts, however, the unit of analysis is a single underlying and its own
complex; spillovers enter, if at all, as bias to be signed (in
\cite{zhao2026}, contaminated controls make the treatment estimate a
lower bound). Our contribution is to promote the coupling from nuisance
to object: we ask what the \emph{system} of interacting mandates does,
and what a per-product audit of it must miss.

\subsection{Cross-impact and coupled microstructure}

That order flow in one asset moves prices of others is documented
in~\cite{hasbrouck2001} and modeled in~\cite{capponi2021}, who attribute
observed cross-impact largely to common components in order flow. The
closing auction is a natural coupling locus: it concentrates index-linked
and hedging flow~\cite{bogousslavsky2023} and is a known target of
reference-price manipulation~\cite{hillion2004}. Our fixed point
\eqref{eq:fixedpoint} embeds a cross-impact matrix inside a feedback loop
whose forcing element---the mandated rebalance---is publicly sized, a
structure absent from the cross-impact literature. On the feedback side,
destabilizing positive-feedback demand goes back
to~\cite{delong1990,gennotte1990}, with fire-sale
spirals~\cite{shleifer2011}, leverage cycles~\cite{adrian2010}, and
predatory amplification~\cite{brunnermeier2005} as the canonical
mechanisms. These literatures treat one self-reinforcing trader class in
one asset; the matrix structure of interacting mandates is, to our
knowledge, unexplored.

\subsection{Signal processing and market microstructure}

The correspondence between financial engineering and signal processing is
developed systematically in~\cite{feng2016}, including order execution as
optimal control; the execution models
of~\cite{almgren2001,bertsimas1998,obizhaeva2013} and the informed-trading
benchmark of~\cite{kyle1985,grossman1988} underlie our impact
parameterization. Those models are open loop: the order schedule is an
exogenous decision variable. The LETF mandate closes the loop, because
the order is a disclosed function of the price being formed. Estimating
the resulting system from outputs alone is a closed-loop identification
problem~\cite{ljung1999}; the stability language---loop gain, stability
margin, spectral radius---is that of multivariable
feedback~\cite{astrom2008}. Our reduced-form moment exploits a physical
property of the loop (overnight reversion of displacement) in place of an
external excitation signal.

\section{Coupled Rebalancing Feedback}\label{sec:model}

Table~\ref{tab:notation} collects the notation used throughout.

\begin{table}[t]
\caption{Notation.}
\label{tab:notation}
\centering
\footnotesize
\setlength{\tabcolsep}{4pt}
\begin{tabular}{@{}ll@{}}
\toprule
Symbol & Meaning \\
\midrule
$n$ & number of coupled underlyings \\
$A_f,\ L_f$ & assets and leverage multiple of fund $f$ \\
$K_i$ & rebal.\ capital $\sum_{f\in i}A_f(L_f^2-L_f)$ of complex $i$ \\
$\gamma_i$ & scaled capital $K_i/\mathrm{ADV}_i$ \\
$\boldsymbol{\Phi}=[\phi_{ij}]$ & own/cross closing-impact coefficients \\
$\mathbf{L}=\boldsymbol{\Phi}\,\mathrm{diag}(\boldsymbol{\gamma})$ & loop-gain matrix; $\ell_{ij}$ its entries \\
$\rL$ & spectral radius; $1-\rL$ = stability margin \\
$\mathbf{M}=(\mathbf{I}-\mathbf{L})^{-1}-\mathbf{I}$ & displacement-response matrix \\
$r_{1,i},\ r_{2,i}$ & day return; closing-window return \\
$r_{\mathrm{on},i}$ & overnight (close-to-next-open) return \\
$\theta$ & overnight correction share of displacement \\
$\beta_\times$ & cross-reversal coefficient in \eqref{eq:reg} \\
$z$ & difference-in-differences statistic \\
\bottomrule
\end{tabular}
\end{table}

\subsection{Setup}

We first record the rebalancing arithmetic that generates the feedback.
A fund promising $L$ times the daily simple return, with net asset value
$A$ and day return $r$, must end the day holding exposure $LA(1+Lr)$
(its post-rebalance target) while its unrebalanced exposure has drifted
to $LA(1+r)$; the mandated trade is the difference,
\begin{equation}
\underbrace{LA(1+Lr)}_{\text{required}} -
\underbrace{LA(1+r)}_{\text{current}} \;=\; A\,(L^{2}-L)\,r ,
\label{eq:arith}
\end{equation}
which is momentum-directional for every $L\notin[0,1]$: a $+2\times$ fund
carries $L^2-L=2$ and a $-2\times$ fund carries $L^2-L=6$, so leveraged
and inverse products \emph{add}~\cite{cheng2009,zhao2026}.
Equation~\eqref{eq:arith} is exact in simple returns; our empirical work
uses log returns, for which it holds to first order---a gap that is
second order in $|r|$ and worth remembering on the $\pm10\%$ days the
Korean sample contains. Aggregating
\eqref{eq:arith} over all products referencing underlying $i$ defines the
complex's rebalancing capital and its liquidity-scaled version,
\begin{equation}
K_i=\sum_{f\in i} A_f\,(L_f^2-L_f), \qquad
\gamma_i = K_i/\mathrm{ADV}_i ,
\label{eq:K}
\end{equation}
both computable daily from public filings.

Consider now $n$ underlyings.
Let $r_{1,i}$ denote the day return on which complex $i$'s mandate is
sized, and $r_{2,i}$ the closing-window return in which mandates execute.
Impact propagates within and across assets through
$\boldsymbol{\Phi}=[\phi_{ij}]\ge 0$: $\phi_{ii}$ is own closing impact
per unit of scaled flow and $\phi_{ij}$ the cross-impact on asset $i$ of
flow executed in asset $j$, nonzero whenever the two assets share
arbitrage capital, index membership, or market-making
inventory~\cite{capponi2021,hasbrouck2001}. With
$\gamma_j = K_j/\mathrm{ADV}_j$ and
$\mathbf{L}=\boldsymbol{\Phi}\,\mathrm{diag}(\boldsymbol{\gamma})$, the
closing price is a fixed point because each mandate is sized on the
\emph{full} daily return $r_{1,j}+r_{2,j}$, i.e., at the very price the
mandate's execution displaces~\cite{zhao2026}. Under (A1) linear impact
over the relevant flow range, (A2) mandates executed against the closing
book, and (A3) elementwise nonnegative coupling
$\boldsymbol{\Phi}\ge 0$, aggregate displacement solves
$\mathbf{r}_2 = \mathbf{L}(\mathbf{r}_1+\mathbf{r}_2)+\mathbf{v}$, i.e.,
\begin{equation}
\mathbf{r}_2 \;=\; (\mathbf{I}-\mathbf{L})^{-1}
\big(\mathbf{L}\,\mathbf{r}_1 + \mathbf{v}\big),
\label{eq:fixedpoint}
\end{equation}
where $\mathbf{v}$ is closing noise. The fixed point exists and is unique
whenever $1\notin\sigma(\mathbf{L})$; we work throughout on the
economically relevant stable region $\rL<1$, where the
distance $1-\rL$ is the \emph{system stability margin}. The multiplier
has the Neumann-series reading familiar from feedback
analysis~\cite{astrom2008},
\begin{equation}
(\mathbf{I}-\mathbf{L})^{-1} \;=\; \sum_{k=0}^{\infty} \mathbf{L}^{k},
\qquad \text{convergent iff } \rL<1,
\label{eq:neumann}
\end{equation}
in which $\mathbf{L}^{k}$ prices the $k$-th round-trip of the echo: a
displacement in asset $j$ enlarges mandates, whose execution displaces
asset $i$, which enlarges mandates again. Off-diagonal terms appear from
$k{=}1$ onward, so cross-echoes are first order even when each product's
own gain is modest. Fig.~\ref{fig:block} draws the system as a feedback
block diagram: the mandate block $\mathrm{diag}(\boldsymbol{\gamma})$
sizes flows from returns, the venue block $\boldsymbol{\Phi}$ maps flows
to price displacement, and the defining feature of the LETF contract---%
sizing at the very price being formed---closes the loop. Coupling lives
entirely in the off-diagonal entries of $\boldsymbol{\Phi}$; a
per-product monitor sees only the two shaded self-loops.

\begin{figure}[t]
\centering
\begin{tikzpicture}[>=Stealth, node distance=6mm and 9mm, font=\scriptsize,
  blk/.style={draw, rounded corners=1pt, minimum width=13mm, minimum height=6mm, fill=blue!6},
  sum/.style={draw, circle, inner sep=1pt}]
\node[sum] (s1) {$+$};
\node[blk, right=7mm of s1] (g1) {$\gamma_1$};
\node[blk, right=of g1] (p11) {$\phi_{11}$};
\node[sum, right=7mm of p11] (o1) {$+$};
\node[sum, below=11mm of s1] (s2) {$+$};
\node[blk, right=7mm of s2] (g2) {$\gamma_2$};
\node[blk, right=of g2] (p22) {$\phi_{22}$};
\node[sum, right=7mm of p22] (o2) {$+$};
\node[left=5mm of s1] (r1) {$r_{1,1}$};
\node[left=5mm of s2] (r2) {$r_{1,2}$};
\node[right=5mm of o1] (y1) {$r_{2,1}$};
\node[right=5mm of o2] (y2) {$r_{2,2}$};
\draw[->] (r1) -- (s1); \draw[->] (s1) -- (g1);
\draw[->] (g1) -- node[above]{flow$_1$} (p11); \draw[->] (p11) -- (o1);
\draw[->] (o1) -- (y1);
\draw[->] (r2) -- (s2); \draw[->] (s2) -- (g2);
\draw[->] (g2) -- node[below]{flow$_2$} (p22); \draw[->] (p22) -- (o2);
\draw[->] (o2) -- (y2);
\draw[->, red!70!black] (g1.east) ++(1.5mm,0) to[out=-45,in=135]
  node[pos=0.35, right=0pt]{$\phi_{21}$} (o2.north west);
\draw[->, red!70!black] (g2.east) ++(1.5mm,0) to[out=45,in=-135]
  node[pos=0.35, right=0pt]{$\phi_{12}$} (o1.south west);
\draw[->, dashed] (o1.north) -- ++(0,4mm) -| node[pos=0.25, above]{sized at close} (s1.north);
\draw[->, dashed] (o2.south) -- ++(0,-4mm) -| (s2.south);
\end{tikzpicture}
\caption{The coupled rebalancing loop as a block diagram
($n{=}2$).}
\vspace{2pt}
{\footnotesize Note: mandates $\mathrm{diag}(\boldsymbol{\gamma})$ size
flows from returns; the venue $\boldsymbol{\Phi}$ maps flows to closing
displacement; dashed feedback arises because the mandate is sized at the
price it displaces. A per-product monitor observes only the top or
bottom row; the red cross-impact channels $\phi_{12},\phi_{21}$---and
every mixed term of \eqref{eq:neumann}---are outside its view.}
\label{fig:block}
\end{figure}
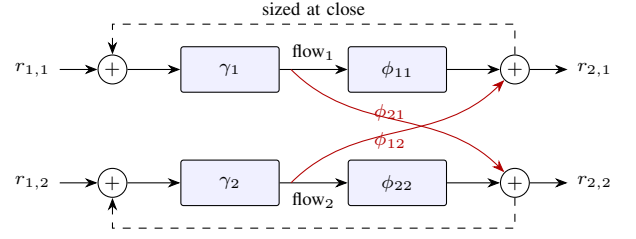

Write
\begin{equation}
\mathbf{M} \;\equiv\; (\mathbf{I}-\mathbf{L})^{-1}-\mathbf{I}
\label{eq:M}
\end{equation}
for the displacement-response matrix: $M_{ij}$ is the closing
displacement of asset $i$ per unit of asset $j$'s day return, net of
fundamentals.

\subsection{Two blind spots of scalar monitoring}

\begin{proposition}[Cycle amplification]\label{prop:rho}
Let $\mathbf{L}\ge 0$ elementwise. Then $\rL \ge \max_i \ell_{ii}$.
For $n{=}2$,
\begin{equation}
\rL = \tfrac{\ell_{11}+\ell_{22}}{2} +
\sqrt{\Big(\tfrac{\ell_{11}-\ell_{22}}{2}\Big)^{2} +
\ell_{12}\ell_{21}},
\label{eq:rho2}
\end{equation}
so the excess over $\max_i\ell_{ii}$ is strictly positive if and only if
$\ell_{12}\ell_{21}>0$, and is increasing in the coupling product.
\end{proposition}

\begin{proof}
The bound is the standard consequence of monotonicity of the spectral
radius for nonnegative matrices: $\mathbf{L}\ge
\mathrm{diag}(\ell_{11},\dots,\ell_{nn})$ elementwise implies $\rL \ge
\rho(\mathrm{diag}) = \max_i\ell_{ii}$~\cite[Ch.~8]{horn2013}. For
$n{=}2$, \eqref{eq:rho2} is the larger root of the characteristic
polynomial $\lambda^2 - (\ell_{11}+\ell_{22})\lambda +
(\ell_{11}\ell_{22}-\ell_{12}\ell_{21})$; the discriminant exceeds
$(\ell_{11}-\ell_{22})^2/4$ exactly when $\ell_{12}\ell_{21}>0$.
\end{proof}

A supervisor watching each $\ell_{ii}$ therefore never overstates, and
generically understates, the system gain. By \eqref{eq:rho2} the excess
behaves as $\ell_{12}\ell_{21}/|\ell_{11}-\ell_{22}|$ for unequal
diagonals and as $\sqrt{\ell_{12}\ell_{21}}$ in the symmetric case, so
the same coupling budget produces far more cycle amplification for
\emph{symmetric} complexes on a common risk factor than for strongly
lopsided ones. Fig.~\ref{fig:theory} plots the excess
$\rL-\max_i\ell_{ii}$ against the coupling product for three
configurations: the asymmetric Korean diagonal, a symmetric
moderate-gain system, and a symmetric U.S.-like system. The same coupling
budget buys an order of magnitude more cycle amplification in the
symmetric case---the geometry behind our empirical split between
Sections~\ref{sec:korea} and~\ref{sec:us}.

\begin{figure}[t]
\centering
\includegraphics[width=0.95\columnwidth]{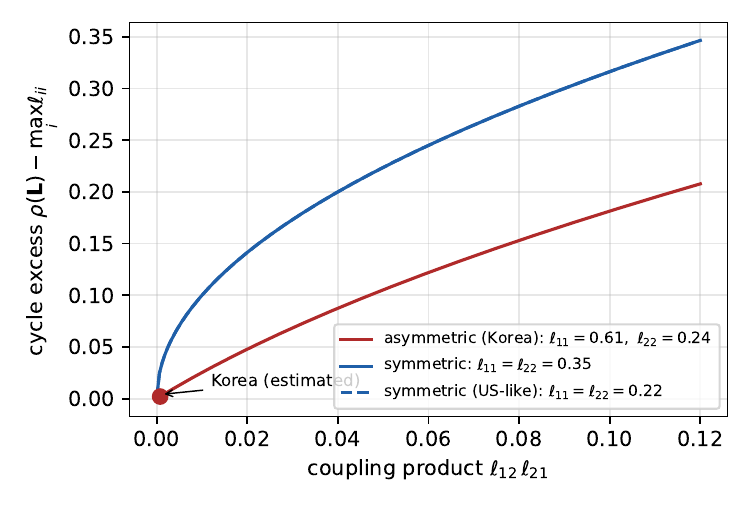}
\caption{Cycle excess $\rL-\max_i\ell_{ii}$ as a function of the coupling
product $\ell_{12}\ell_{21}$, from the closed form \eqref{eq:rho2}. The
excess is largest for symmetric diagonals; the marker locates the Korean
point estimate (headline variant of Table~\ref{tab:quant}), whose
one-directional coupling places it near the horizontal axis.}
\label{fig:theory}
\end{figure}

\begin{proposition}[Transmitted displacement]\label{prop:transmit}
Let $n{=}2$ with $\ell_{11},\ell_{22}\in[0,1)$ and one-way coupling:
$\ell_{21}>0$, $\ell_{12}=0$. Then
$\rL=\max(\ell_{11},\ell_{22})$, yet
\begin{equation}
M_{21} \;=\; \frac{\ell_{21}}{(1-\ell_{11})(1-\ell_{22})} \;>\; 0 ,
\label{eq:M21}
\end{equation}
while a per-product monitor of asset 2 observes only $\ell_{22}$ and its
own multiplier $1/(1-\ell_{22})$.
\end{proposition}

\begin{proof}
For triangular $\mathbf{L}$ the spectrum is the diagonal. Direct
inversion of $\mathbf{I}-\mathbf{L}$ for the lower-triangular case gives
the off-diagonal entry of $(\mathbf{I}-\mathbf{L})^{-1}$ as
$\ell_{21}[(1-\ell_{11})(1-\ell_{22})]^{-1}$, which equals $M_{21}$
by \eqref{eq:M}.
\end{proof}

The two propositions separate two failure modes of product-level
surveillance. Cycle amplification raises the \emph{system} gain and thus
the distance to instability; it requires a two-way cycle. Transmitted
displacement leaves the spectral radius untouched but injects another
complex's displacement---sized by the \emph{sender's} capital---into the
receiver's reference price; one direction suffices. Both are invisible to
scalar audits by construction.

\subsection{Reduced-form identification}\label{sec:ident}

Displacement reverts as outside capital corrects the close overnight.
Let $\theta\in(0,1]$ be the overnight correction share and
$\mathbf{r}_{\mathrm{on}}$ the vector of close-to-next-open returns.

\begin{proposition}[Identification]\label{prop:ident}
Suppose $\mathbf{r}_{\mathrm{on}} = -\theta(\mathbf{r}_2 -
\mathbb{E}[\mathbf{r}_2\,|\,\text{fundamentals}]) + \mathbf{w}$, and
assume the non-displacement components---$\mathbf{v}$
in \eqref{eq:fixedpoint}, $\mathbf{w}$, and the fundamental part of
$\mathbf{r}_2$---are orthogonal to $\mathbf{r}_1$. Then the
coefficient matrix $\mathbf{B}$ of the multivariate regression of
$\mathbf{r}_{\mathrm{on}}$ on $\mathbf{r}_1$ satisfies
$\mathbf{B}=-\theta\,\mathbf{M}$, and
\begin{equation}
\mathbf{L} \;=\; \mathbf{I} - (\mathbf{I}+\mathbf{M})^{-1},
\qquad \mathbf{M} = -\mathbf{B}/\theta .
\label{eq:invert}
\end{equation}
\end{proposition}

\begin{proof}
By \eqref{eq:fixedpoint} and \eqref{eq:M},
$\mathbf{r}_2=(\mathbf{I}+\mathbf{M})(\mathbf{L}\mathbf{r}_1+\mathbf{v})
= \mathbf{M}\mathbf{r}_1 + (\mathbf{I}+\mathbf{M})\mathbf{v}$ up to the
fundamental component, using
$(\mathbf{I}-\mathbf{L})^{-1}\mathbf{L}=\mathbf{M}$. Independence of
$\mathbf{v},\mathbf{w}$ from $\mathbf{r}_1$ gives
$\mathbb{E}[\mathbf{r}_{\mathrm{on}}|\mathbf{r}_1] =
-\theta\,\mathbf{M}\,\mathbf{r}_1$. The map
$\mathbf{L}\mapsto\mathbf{M}$ is inverted by
$\mathbf{I}+\mathbf{M}=(\mathbf{I}-\mathbf{L})^{-1}$.
\end{proof}

Three qualifications delimit the proposition. \emph{(i) Late news.} If
fundamental news $\mathbf{g}$ arrives during the closing window,
$\mathbf{r}_2$ gains a component that does not revert; provided
$\mathbf{g}$ is independent of $\mathbf{r}_1$, it enters the error and
leaves $\mathbf{B}$ unbiased, at a cost in precision only. \emph{(ii)
Heterogeneous correction.} With asset-specific correction shares
$\boldsymbol{\Theta}=\mathrm{diag}(\theta_i)$, the map becomes
$\mathbf{M}=-\boldsymbol{\Theta}^{-1}\mathbf{B}$, identified when each
$\theta_i$ is calibrated per asset; we use a common $\theta$ and report
sensitivity. \emph{(iii) Time variation.} With
$\mathbf{L}_t=\boldsymbol{\Phi}\,\mathrm{diag}(\boldsymbol{\gamma}_t)$,
the regression estimand is the variance-weighted average of
$\mathbf{M}_t$ over the estimation window; the $\gamma$-interaction
below targets the variation directly. \emph{(iv) Measurement
convention.} The model's $\mathbf{r}_1$ is the pre-close return, but
our empirical regressor is the close-to-close day return
$\mathbf{R}=(\mathbf{I}+\mathbf{M})(\mathbf{r}_1+\mathbf{v})$. Under
joint normality the exact estimand of the $\mathbf{R}$-regression is
\begin{equation}
\mathbf{B}=-\theta\big[\mathbf{M}\boldsymbol{\Sigma}_1
+(\mathbf{I}+\mathbf{M})\boldsymbol{\Sigma}_v\big]
(\boldsymbol{\Sigma}_1+\boldsymbol{\Sigma}_v)^{-1}
(\mathbf{I}+\mathbf{M})^{-1},
\label{eq:estimand}
\end{equation}
with $\boldsymbol{\Sigma}_1,\boldsymbol{\Sigma}_v$ the fundamental and
closing-noise covariances (verified numerically in the replication
code). Because the covariance split is unobserved, we do not claim the
two polar cases of \eqref{eq:estimand} bound the truth; they are the
natural \emph{reporting conventions}---the $\mathbf{L}$-map, exact as
$\boldsymbol{\Sigma}_v\to 0$ since
$\mathbf{M}(\mathbf{I}+\mathbf{M})^{-1}=\mathbf{L}$, and the
$\mathbf{M}$-map, the approximation $\mathbf{R}\approx\mathbf{r}_1$
appropriate when the closing window is a small share of day
variance---and Section~\ref{sec:quant} reports both as a sensitivity
analysis. The \emph{detection} results
(sign, DiD, placebo, randomization) are reduced-form contrasts invariant
to this choice. Under these conditions, with
stationary and ergodic regressors whose second-moment matrix is finite
and nonsingular, finite fourth moments within each regime, and the
error components of qualification~(iv) orthogonal to the regressors,
$\hat{\mathbf{B}}$ is consistent by
standard arguments; we
nevertheless base small-sample inference on exact randomization
(Section~\ref{sec:korea}) rather than on asymptotics.

Three features do the identification work in practice. First, the moment
uses \emph{cross} reversals, which the scalar literature discards.
Second, coupling implies a \emph{negative} cross-coefficient
(displacement reverts), whereas the leading confounder---fundamental
lead--lag among correlated equities~\cite{lo1990,chordia2000}---implies a
positive one (news continues); the sign is an orthogonal discriminator.
Third, $K_{j,t}$ is observed daily, so the coupling loading must scale
with the sender's capital, a restriction a static lead--lag cannot mimic.
$\theta$ is not separately identified from a single cross-section; we
import it from external calibration ($0.88$ in~\cite{zhao2026}) and
report sensitivity over $\theta\in\{0.7,0.88,1.0\}$
(Table~\ref{tab:quant}). Two further properties of the moment follow
directly: statistical power is proportional to $\theta$ (a market whose
displacements are never corrected overnight leaves no reversal to
measure), and the signal grows as the system approaches instability,
since $\|\mathbf{M}\|\to\infty$ as $\rL\to 1$---the estimator is most
informative exactly where monitoring matters most.

\begin{remark}
Estimation of $\rho$ through \eqref{eq:invert} inherits heavy tails on
paths where $\mathbf{I}+\hat{\mathbf{M}}$ is near singular; monitoring
implementations should regularize the inversion. False-alarm rates
reported below are unaffected.
\end{remark}

Algorithms~\ref{alg:est} and~\ref{alg:mon} summarize the procedure in the
form in which an exchange or supervisor would run it. Every input is
public: prices, fund filings, and venue traded values.

\begin{algorithm}[t]
\caption{Reduced-form estimation of $\mathbf{L}$ and $\rL$}
\label{alg:est}
\begin{algorithmic}[1]
\REQUIRE daily open/close prices of the $n$ underlyings; fund shares
outstanding, NAVs, and multiples $L_f$; market and
overnight-window controls; correction share $\theta$
\FOR{each trading day $t$}
  \STATE $r_{1,i}(t) \leftarrow \ln C_{i,t}/C_{i,t-1}$;\quad
         $r_{\mathrm{on},i}(t) \leftarrow \ln O_{i,t+1}/C_{i,t}$
  \STATE $K_{i,t} \leftarrow \sum_{f\in i} A_{f,t-1}(L_f^2-L_f)$;\quad
         $\gamma_{i,t} \leftarrow K_{i,t}/\mathrm{ADV}_{i,t-1}$
         \COMMENT{lagged; no look-ahead}
\ENDFOR
\FOR{each receiver $i=1,\dots,n$}
  \STATE regress $r_{\mathrm{on},i}$ on
         $\{r_{1,j}\}_{j=1}^{n}$ and controls (OLS, robust s.e.)
  \STATE $\hat B_{ij} \leftarrow$ coefficient on $r_{1,j}$
\ENDFOR
\STATE $\hat{\mathbf{M}} \leftarrow -\hat{\mathbf{B}}/\theta$;\quad
       $\hat{\mathbf{L}} \leftarrow \mathbf{I} -
       (\mathbf{I}+\hat{\mathbf{M}}+\epsilon\mathbf{I})^{-1}$
       \COMMENT{$\mathbf{M}$-map; report $\hat{\mathbf{L}}=-\hat{\mathbf{B}}/\theta$
       ($\mathbf{L}$-map) alongside, per qualification~(iv);
       $\epsilon>0$ only if ill-conditioned}
\STATE \textbf{sign screen:} accept coupling $\hat\ell_{ij}$ only if
       $\hat B_{ij}<0$ (reversal, not lead--lag continuation)
\STATE \textbf{scale screen:} corroborate via interaction of $r_{1,j}$
       with observed $\gamma_{j,t}$
\RETURN $\hat{\mathbf{L}}$, $\rho(\hat{\mathbf{L}})$, margins
        $1-\rho(\hat{\mathbf{L}})$
\end{algorithmic}
\end{algorithm}

\begin{algorithm}[t]
\caption{Matrix stability monitoring}
\label{alg:mon}
\begin{algorithmic}[1]
\REQUIRE window length $W$; thresholds $\bar\rho$ (system),
$\bar m$ (transmission); consecutive-exceedance count $c$
\FOR{each day $t$}
  \STATE run Algorithm~\ref{alg:est} on the trailing window
         $\{t-W+1,\dots,t\}$
  \STATE $s_1(t) \leftarrow \rho(\hat{\mathbf{L}}_t)$
         \COMMENT{cycle statistic}
  \STATE $s_2(t) \leftarrow \max_{i\ne j} \hat M_{ij,t}$
         \COMMENT{transmission statistic}
  \IF{$s_1 > \bar\rho$ \OR $s_2 > \bar m$ for $c$ consecutive windows}
    \STATE raise alarm; report the offending (complex, venue) cell
  \ENDIF
\ENDFOR
\end{algorithmic}
\end{algorithm}

\subsection{Dynamic extension: chaining trading days}\label{sec:dynamic}

The fixed point \eqref{eq:fixedpoint} is static within the day. Chaining
days couples the displacements: if a share $\theta$ of the closing error
$\mathbf{e}_t$ (the displacement vector) is corrected overnight, the
uncorrected share $(1-\theta)\mathbf{e}_t$ persists in the next close
while the day's mandate responds to the priced gap
$\boldsymbol{\varepsilon}_{t+1}-\theta\mathbf{e}_t$; substituting into
\eqref{eq:fixedpoint} gives the vector recursion
\begin{equation}
\mathbf{e}_{t+1} = \big[(1-\theta)\mathbf{I}-\theta\mathbf{M}\big]
\mathbf{e}_t
+ \mathbf{M}\,\boldsymbol{\varepsilon}_{t+1}
+ (\mathbf{I}+\mathbf{M})\,\mathbf{v}_{t+1},
\label{eq:chain}
\end{equation}
whose scalar reduction has autoregressive coefficient
$1-\theta(1+M)=1-\theta\Pi_g$, exactly the error chain
of~\cite{zhao2026}---a useful cross-check on the derivation.

\begin{proposition}[Dynamic stability]\label{prop:dynamic}
Let the innovations in \eqref{eq:chain} be i.i.d.\ with finite,
nondegenerate covariance, and consider the causal (stationary-solution)
initialization. The system admits a covariance-stationary solution if
and only
if $\rho\big((1-\theta)\mathbf{I}-\theta\mathbf{M}\big)<1$. If
$\mathbf{L}$ has real spectrum in $[0,1)$, this is equivalent to
\begin{equation}
\rL \;<\; 1-\frac{\theta}{2},
\label{eq:dynthresh}
\end{equation}
a strictly tighter bound than the static existence condition $\rL<1$,
and every spectral direction with $\lambda_i>1-\theta$ approaches
instability through alternating-sign oscillation.
\end{proposition}

\begin{proof}
Stationarity of the VAR(1) \eqref{eq:chain} requires the coefficient's
spectral radius below one. $\mathbf{M}$ is the rational function
$\lambda\mapsto\lambda/(1-\lambda)$ of $\mathbf{L}$, so by spectral
mapping the coefficient has eigenvalues
$(1-\theta)-\theta\lambda_i/(1-\lambda_i)=1-\theta/(1-\lambda_i)$. For
$\lambda_i\in[0,1)$ and $\theta\in(0,1]$, the condition
$|1-\theta/(1-\lambda_i)|<1$ is $0<\theta/(1-\lambda_i)<2$, i.e.,
$\lambda_i<1-\theta/2$; the eigenvalue is negative---oscillation---%
exactly when $\theta/(1-\lambda_i)>1$, i.e., $\lambda_i>1-\theta$.
\end{proof}

Proposition~\ref{prop:dynamic} sharpens the monitoring message: the
operative threshold is not $\rL=1$ but the dynamic bound
\eqref{eq:dynthresh}---$0.56$ at the calibrated $\theta=0.88$. The
SK~Hynix scalar gain of $0.61$ already exceeds it, consistent with the
oscillatory, boundary-hitting behavior reported for the episode
in~\cite{zhao2026}; the matrix statistics of Algorithm~\ref{alg:mon}
should therefore be read against \eqref{eq:dynthresh} rather than
against unity.

\section{Synthetic Validation}\label{sec:synth}

We simulate the two-asset system \eqref{eq:fixedpoint} with a common
fundamental factor (loadings $1.0$ and $0.8$, so day returns are
correlated without any coupling), $T{=}250$ trading days, 500 Monte Carlo
paths, and overnight reversal
$\mathbf{r}_{\mathrm{on}}=-\theta\,\mathbf{r}_2+\mathbf{w}$. Five
scenarios were fixed before estimation and all are reported:
(a)~symmetric coupling with $\rL=0.60$ and constant $\gamma$ (clean
recovery); (b1)~no coupling, common factor only (false-alarm control);
(b2)~no coupling, fundamental lead--lag of $0.25$ from asset~1 into
asset~2's next-day return (the dangerous confounder); (b3)~coupling and
lead--lag simultaneously (power under contamination);
(c)~own-gains $0.35$ each with coupling $0.30$, hence $\rL=0.65$ (the
blind-spot configuration). In scenarios with time-varying $\gamma_t$
(persistent lognormal, 30\% relative dispersion) we additionally test the
capital-scaling discriminator by interacting the sender's return with the
observed $\gamma_{j,t}$.

\begin{table}[t]
\caption{Synthetic validation of the reduced-form estimator.}
\vspace{2pt}
\label{tab:synth}
\centering
\footnotesize
\setlength{\tabcolsep}{3pt}
\begin{tabular}{@{}lcc@{}}
\toprule
Scenario & Statistic & Value \\
\midrule
(a) coupled, $\rL=0.60$ & $\hat\ell_{11}$ bias / RMSE & $0.000$ / $0.020$ \\
 & $\hat\ell_{12}$ bias / RMSE & $0.000$ / $0.022$ \\
 & $\hat\rho$ RMSE & $\mathbf{0.005}$ \\
 & coupling detection ($|t|>2$) & $1.00$ \\
(b1) factor only & false alarm & $0.06$ \\
(b2) lead--lag only & false alarm & $\mathbf{0.02}$ \\
(b3) coupling $+$ lead--lag & detection & $0.78$ \\
(c) blind spot, $\rL=0.65$ & scalar monitor ``unsafe'' & $0.00$ \\
 & matrix monitor ``unsafe'' & $\mathbf{1.00}$ \\
\midrule
$n{=}3$ (a3) coupled, $\rL=0.60$ & $\hat\rho$ RMSE & $0.004$ \\
$n{=}3$ (c3) blind spot & scalar / matrix ``unsafe'' & $0.00$ / $\mathbf{1.00}$ \\
$n{=}3$ (b3) factor only & off-diag.\ flag (familywise) & $0.14$ \\
\bottomrule
\end{tabular}

\vspace{3pt}
{\footnotesize Note: $T{=}250$ days, 500 MC paths per scenario (300 for
$n{=}3$); ``unsafe'' threshold $0.5$ for both statistics; detection and
false alarm here use the two-sided rule $|t|>2$, whereas
Table~\ref{tab:baseline} applies the one-sided reversal-sign rule, which
is why the false-alarm figures differ slightly. The $n{=}3$ familywise
flag rate (any of six off-diagonals) corresponds to ${\sim}2.4\%$ per
pair.}
\end{table}

\begin{figure}[t]
\centering
\includegraphics[width=0.95\columnwidth]{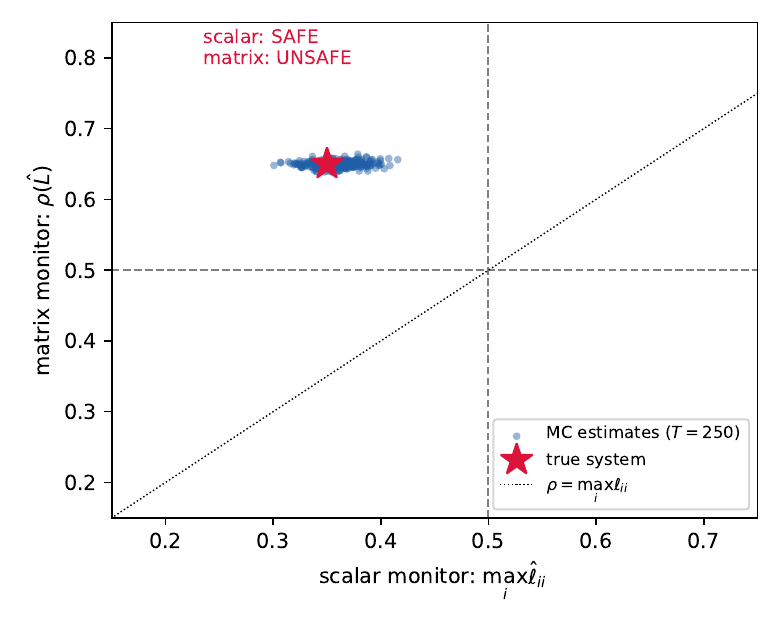}
\caption{The blind spot. Each dot is one Monte Carlo path of the
calibrated configuration: own-gains $\ell_{11}=\ell_{22}=0.35$, coupling
$\ell_{12}=\ell_{21}=0.30$, hence $\rL=0.65$ (star). Against a threshold
of $0.5$ (dashed), the scalar statistic $\max_i\hat\ell_{ii}$ clears every
path as safe while $\rho(\hat{\mathbf{L}})$ flags every path as unsafe.}
\label{fig:blindspot}
\end{figure}

Table~\ref{tab:synth} summarizes. The estimator recovers the matrix and
its spectral radius essentially without bias; the lead--lag confounder is
rejected by the sign discriminator (2\% false alarms at nominal 5\%);
simultaneous coupling and lead--lag attenuate power (0.78) without
reversing conclusions. The $\gamma$-interaction discriminator is weaker
in simulation (false-alarm 22--27\% under time-varying $\gamma$ without
coupling) and is therefore used as corroboration, not as the primary
test. Scenario (c) is the paper's central exhibit
(Fig.~\ref{fig:blindspot}): a system whose every product passes a scalar
audit while the system gain exceeds the same threshold, detected on every
path by the matrix statistic.

\subsection{Comparison with alternative detectors and measurements}

Two comparisons discipline the design. First, against a \emph{naive
cross-predictability detector}---the natural first instinct, flagging
coupling whenever asset $j$'s return predicts asset $i$'s subsequent
return, without controls or a sign restriction---our detector's
advantage is entirely in false alarms (Table~\ref{tab:baseline}):
the naive rule fires on 100\% of no-coupling paths, because a common
factor or an ordinary lead--lag~\cite{lo1990,chordia2000} is
indistinguishable from coupling once the reversal sign and the
conditioning set are discarded. Detection power at true coupling is
identical. The comparison quantifies why the older spillover literature
could not have isolated this channel with unsigned cross-correlations.

\begin{table}[t]
\caption{Detector comparison on the synthetic scenarios.}
\vspace{2pt}
\label{tab:baseline}
\centering
\small
\setlength{\tabcolsep}{4pt}
\begin{tabular}{@{}lcc@{}}
\toprule
Scenario (300 paths) & ours & naive \\
\midrule
coupled ($\rho=0.60$): detection & 1.00 & 1.00 \\
common factor only: false alarm & 0.05 & \textbf{1.00} \\
lead--lag only: false alarm & 0.04 & \textbf{1.00} \\
coupled $+$ lead--lag: detection & 1.00 & 1.00 \\
\bottomrule
\end{tabular}

\vspace{3pt}
{\footnotesize Note: ``ours'' = cross-reversal coefficient with own,
market, and window controls and the reversal-sign restriction
($t<-2$); ``naive'' = unsigned cross-predictability without controls
($|t|>2$).}
\end{table}

Second, against the \emph{direct measurement} of~\cite{zhao2026}, which
estimates the impact curve from signed intraday order flow and multiplies
by capital: the two approaches are complements with different
requirements and outputs (Table~\ref{tab:methods}). Direct measurement
delivers the levels of own-gains with high precision but requires
exchange-internal data and is inherently single-asset and ex post;
the reduced form runs on public data in real time and identifies the
coupled structure, at the cost of importing $\theta$ and a calibrated
diagonal. We use each where it is strong: their diagonal, our
off-diagonal.

\begin{table}[t]
\caption{Measurement approaches compared.}
\vspace{2pt}
\label{tab:methods}
\centering
\small
\setlength{\tabcolsep}{3.5pt}
\begin{tabular}{@{}lcc@{}}
\toprule
 & Direct impact~\cite{zhao2026} & Reduced form (ours) \\
\midrule
data & signed intraday flow & public prices, filings \\
object & scalar $\ell$ per asset & matrix $\mathbf{L}$, $\rL$ \\
coupling & treated as bias & identified \\
timing & ex post (episode) & rolling / real time \\
free inputs & impact curve fit & $\theta$, diagonal calib. \\
precision (own gain) & high & low \\
\bottomrule
\end{tabular}
\end{table}

\section{Evidence I: the Korean Episode}\label{sec:korea}

\subsection{Setting and data}

Sixteen single-stock LETFs on Samsung Electronics and SK~Hynix listed on
May~27,~2026---the largest launch in Korean ETF history. Worldwide
rebalancing capital referencing SK~Hynix peaked near KRW~37.5tn against
KRW~6.6tn for Samsung~\cite{zhao2026}; both complexes rebalance against
the same ten-minute Seoul closing auction, and the pair constitutes
roughly half of index capitalization. Fig.~\ref{fig:episode} shows the
setting: the underlyings rose several-fold in the AI-memory rally,
peaked in late June, whipsawed violently through July, and the flagship
$2\times$ product lost most of its value to volatility decay while the
underlyings ended far above their January levels---the footprint of
repeated overshoot-and-reversal rather than of a directional crash.

\begin{figure}[t]
\centering
\includegraphics[width=0.95\columnwidth]{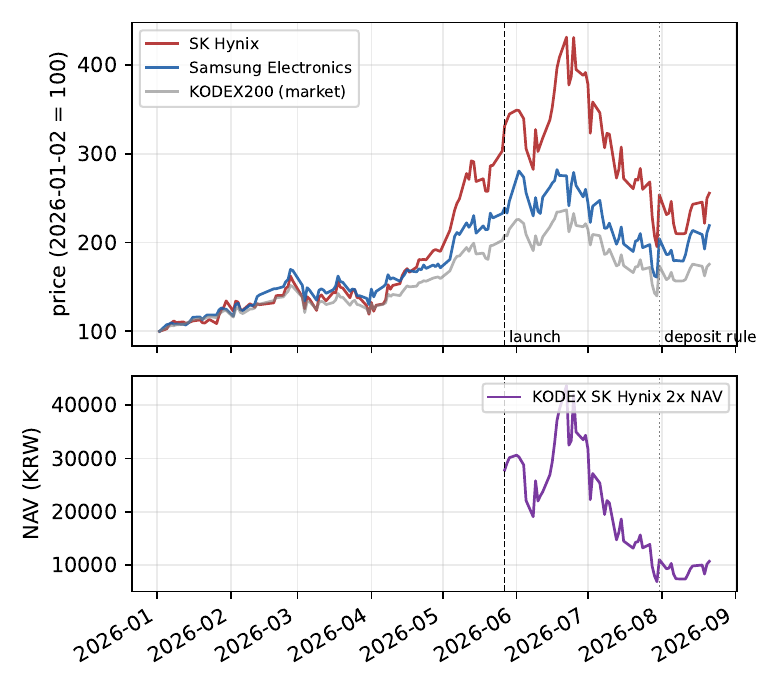}
\caption{The Korean episode, January--August 2026.}
\vspace{2pt}
{\footnotesize Note: top---closing prices indexed to 100 on
January~2,~2026, for the treated pair and the market proxy;
bottom---NAV of the largest single-stock product (KODEX SK~Hynix
$2\times$). Dashed line: single-stock LETF launch (May~27); dotted:
deposit-requirement regulation (July~31).}
\label{fig:episode}
\end{figure}

Daily open/close prices for the
treated pair, fourteen non-treated large capitalization controls, and the
KODEX200 market proxy span January--August 2026 (about 95 pre-launch and
56--60 post-launch trading days). Fund-level listed shares and net asset
values from the exchange's daily statistics give the observed
$K_{j,t}$, hence $\gamma_{j,t}=K_{j,t}/\mathrm{ADV}_{j,t}$ with a 20-day
trailing average of traded value, both lagged one day to avoid look-ahead.
Table~\ref{tab:data} summarizes all data sources.

\begin{table}[t]
\caption{Data sources.}
\label{tab:data}
\centering
\footnotesize
\setlength{\tabcolsep}{4pt}
\begin{tabular}{@{}lll@{}}
\toprule
Series & Source & Period \\
\midrule
KR daily OHLC (16 stocks) & exchange feed & 2026-01--08 \\
KR fund shares/NAV (20 funds) & KRX daily stats & 2026-05--08 \\
KR market proxy (KODEX200) & exchange feed & 2026-01--08 \\
U.S. session control (SMH) & consol.\ tape & 2026-01--08 \\
U.S. funds/stocks (15 tickers) & CRSP daily & 2023--2024 \\
\bottomrule
\end{tabular}
\end{table}

Fig.~\ref{fig:gamma} plots the two domestic complexes' scaled capital:
$\gamma_t$ is of order one throughout the post-launch window---mandated
flow comparable to a full day's traded value per unit
return---and, notably, is \emph{not} extinguished by the July
participation regulations, which collapsed secondary-market turnover by
roughly 90\% while fund assets persisted; participation rules and
loop-gain dials are different instruments.

\begin{figure}[t]
\centering
\includegraphics[width=0.95\columnwidth]{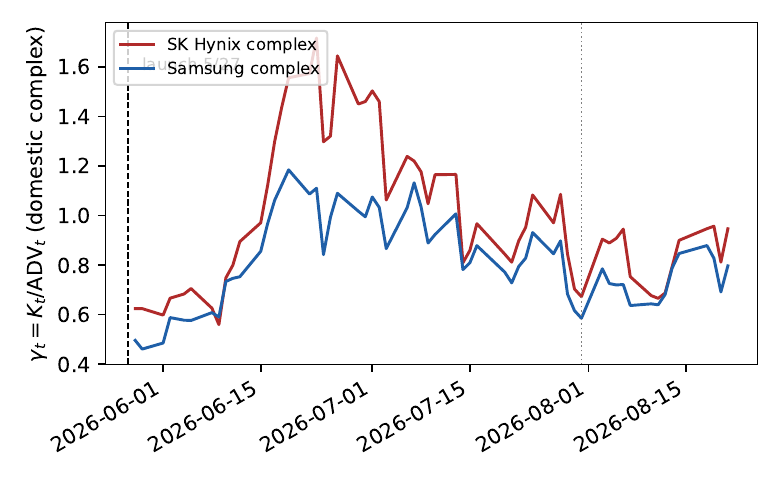}
\caption{Scaled rebalancing capital $\gamma_t$ of the two Korean domestic
complexes (KRX daily listed shares $\times$ NAV over trailing 20-day
traded value, lagged). Dashed line: single-stock LETF launch
(May~27,~2026); dotted line: the deposit-requirement regulation
(July~31), which collapsed turnover but not $\gamma_t$.}
\label{fig:gamma}
\end{figure}

\subsection{Design}

For an ordered pair (sender $A$ $\to$ receiver $B$) we estimate
\begin{multline}
r_{\mathrm{on},B}(t) = \alpha + \beta_\times\, r_{1,A}(t) +
\beta_{\mathrm{own}}\, r_{1,B}(t) \\
+ \mu\, r_{1,\mathrm{mkt}}(t) + \psi\, r_{\mathrm{US}}(t) +
\varepsilon(t),
\label{eq:reg}
\end{multline}
by OLS with heteroskedasticity-robust standard errors; recomputing the
treated statistic with Newey--West HAC errors (5 lags) moves it from
$z=-2.82$ to $z=-2.72$, so residual autocorrelation is not driving the
result, and inference in any case rests on the exact randomization
below. The specification is split before and
after the launch, with the difference-in-differences statistic
$z=(\hat\beta_\times^{\mathrm{post}}-\hat\beta_\times^{\mathrm{pre}})/
(\mathrm{se}^2_{\mathrm{pre}}+\mathrm{se}^2_{\mathrm{post}})^{1/2}$.
Coupling predicts $\beta_\times<0$ after launch for the treated pair
only, with magnitude ordered by the \emph{sender's} capital.

One control is essential and, to our knowledge, unremarked in this
setting: the Korean overnight window (15:30 to next 09:00 KST) contains
the entire U.S.\ trading session, and treated semiconductor stocks load
on it far more heavily than placebo pairs from other sectors. Without the
same-calendar-day U.S.\ semiconductor return $r_{\mathrm{US}}$ (SMH), the
treated cross-coefficients appear \emph{positive}---continuation, not
reversal---and the design would wrongly reject coupling; with the
control, U.S.\ loadings enter with $t\approx 5$ and the artifact is
absorbed. Overnight-reversal designs on Asian markets should treat this
control as mandatory.

\subsection{Results}

\begin{table}[t]
\caption{Cross overnight reversal around the Korean launch.}
\vspace{2pt}
\label{tab:koreagrid}
\centering
\small
\setlength{\tabcolsep}{3pt}
\begin{tabular}{@{}llccc@{}}
\toprule
Direction & window & pre $t$ & post $\hat\beta_\times$ ($t$) & DiD $z$ \\
\midrule
Hynix$\to$Samsung & on & $+3.2$ & $-0.195\ (-1.1)$ & $\mathbf{-2.82}$ \\
 & on, ex.~8/19 & $+3.2$ & $-0.177\ (-1.0)$ & $-2.73$ \\
 & nd & $+1.1$ & $-0.221\ (-0.6)$ & $-1.11$ \\
 & nd, ex.~8/19 & $+1.1$ & $-0.169\ (-0.5)$ & $-1.00$ \\
Samsung$\to$Hynix & on & $+0.7$ & $+0.540\ (+1.4)$ & $+0.93$ \\
 & nd & $-0.2$ & $+1.102\ (+2.1)$ & $+1.95$ \\
Hyundai$\to$Kia & on & $-0.9$ & $+0.010\ (+0.1)$ & $+0.83$ \\
Kia$\to$Hyundai & on & $-1.6$ & $-0.033\ (-0.2)$ & $+0.84$ \\
LGES$\to$SDI & on & $-1.8$ & $-0.157\ (-1.5)$ & $+0.16$ \\
SDI$\to$LGES & on & $+1.8$ & $+0.035\ (+0.5)$ & $-0.79$ \\
\midrule
\multicolumn{5}{@{}l@{}}{$\gamma$-interaction, Hynix$\to$Samsung (post):
$\hat\beta_{\gamma\times} = -0.093$, $t=-2.21$.}\\
\bottomrule
\end{tabular}

\vspace{3pt}
{\footnotesize Note: full pre-registered grid of
spec.~\eqref{eq:reg} with DiD split at May~27,~2026; ``on'' =
overnight receiver window, ``nd'' = next-day close-to-close;
``ex.~8/19'' drops the large U.S.-shock day. Robust standard errors.}
\end{table}

Table~\ref{tab:koreagrid} reports the full pre-registered grid. Three
findings. \emph{First}, the Hynix$\to$Samsung cross-coefficient switches
from positive continuation pre-launch to reversal post-launch in the
overnight window ($z=-2.82$; $-2.73$ excluding August~19, a large
U.S.-shock day). Placebo pairs of the same market and comparable
comovement (Hyundai--Kia, LG Energy Solution--Samsung SDI) show nothing
in any window. \emph{Second}, the loading scales with the sender's
observed capital: interacting $r_{1,A}$ with standardized
$\gamma_{A,t}$ yields $t=-2.21$. \emph{Third}, the effect is
directionally asymmetric exactly as relative complex size predicts:
off-diagonal gain is proportional to the sender's capital, the Hynix
complex is ${\sim}5.7\times$ Samsung's, and the reverse direction is
correspondingly undetectable at this sample size.

A skeptic's alternative deserves its own paragraph: the pre-period
loading is strongly positive ($t=+3.2$), so a negative DiD could in
principle reflect the \emph{decay} of an ordinary lead--lag in the
volatile post-launch regime rather than the arrival of coupling. Three
observations weigh against it. The post-period point estimate is
negative in level, not merely smaller; placebo pairs traded through the
same volatility regime without any shift; and the $\gamma$-interaction
ties the post-period loading to the sender's fund capital within the
post window, a margin a decaying lead--lag does not possess. The
quantification below therefore reports both bases---treating the
pre-period lead--lag as vanished (conservative) or as persisting (upper
bound)---rather than adjudicating what cannot be identified from 56
days. One opposite-signed
cell is recorded honestly: Samsung$\to$Hynix in the next-day window
drifts positive ($z\le+2.16$ across variants), which is continuation
rather than reversal, is absent from placebos, and remains unexplained;
it is not evidence of coupling and we flag it for future data.

\begin{table}[t]
\caption{Inference for the treated statistic.}
\vspace{2pt}
\label{tab:inference}
\centering
\footnotesize
\setlength{\tabcolsep}{3pt}
\begin{tabular}{@{}lc@{}}
\toprule
Exact randomization over control pairs & \\
\quad ordered pairs (14 large caps) & 182 \\
\quad rank of treated $z$ (one-sided) & $0/182$ \\
\quad exact $p$ (one- / two-sided) & $0.0055$ / $0.011$ \\
\quad pre-corr $\ge 0.3$ subset (178) & $p=0.0056$ \\
\quad control-$z$ quantiles (1/50/99\%) & $-1.37$, $+0.65$, $+2.47$ \\
\midrule
Date-block bootstrap (10-day blocks, 1000 draws) & \\
\quad DiD mean & $-0.68$ \\
\quad 95\% interval & $[-1.25,\,-0.15]$ \\
\quad $P(\mathrm{DiD}\ge 0)$ & $0.005$ \\
\bottomrule
\end{tabular}

\vspace{3pt}
{\footnotesize Note: treated statistic = Hynix$\to$Samsung, overnight
window, $z=-2.82$.}
\end{table}

Because only two stocks are treated, clustered asymptotics are
unreliable and we rest inference on exact randomization
(Table~\ref{tab:inference} and Fig.~\ref{fig:randomization}): applying
the identical frozen statistic to all 182 ordered pairs of fourteen
non-treated large caps, no control pair is as negative as the treated
$z$. A ten-day block bootstrap of the DiD excludes zero.

\begin{figure}[t]
\centering
\includegraphics[width=0.95\columnwidth]{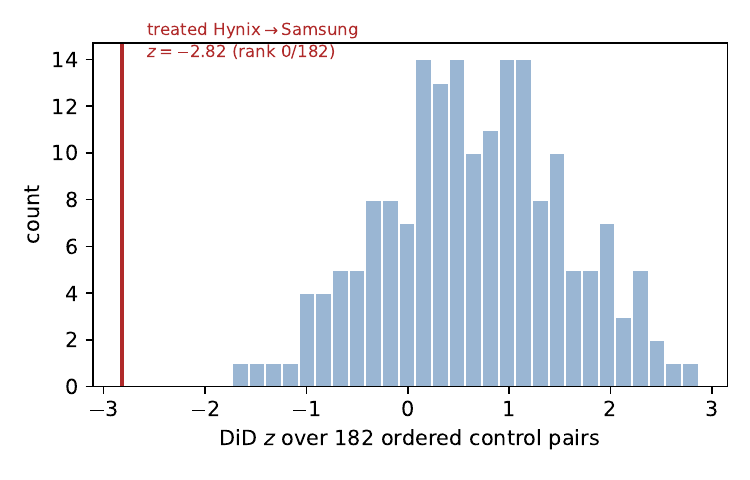}
\caption{Exact randomization distribution. The histogram is the DiD $z$
of the frozen specification applied to every ordered pair of fourteen
non-treated large capitalization stocks (182 pairs); the vertical line is
the treated Hynix$\to$Samsung statistic, more negative than every control
pair (one-sided exact $p=0.0055$).}
\label{fig:randomization}
\end{figure}

One imbalance is recorded: the treated pair's pre-period
return correlation ($0.86$) exceeds the control median ($0.55$), because
very-high-correlation control pairs are structurally scarce; the
high-correlation subset result partially addresses this.

\subsection{Quantification}\label{sec:quant}

\begin{table}[t]
\caption{Assembled $\hat{\mathbf{L}}$: all recorded variants.}
\vspace{2pt}
\label{tab:quant}
\centering
\small
\setlength{\tabcolsep}{4pt}
\begin{tabular}{@{}cccccc@{}}
\toprule
$\theta$ & basis & $M_{21}$ & $\hat\ell_{21}$ & $\rho(\hat{\mathbf{L}})$
 & share \\
\midrule
0.70 & level & 0.279 & 0.083 & 0.613 & 0.52 \\
0.70 & DiD & 0.946 & 0.280 & 0.644 & 0.91 \\
\textbf{0.88} & \textbf{level} & \textbf{0.222} & \textbf{0.066} &
\textbf{0.612} & \textbf{0.41} \\
0.88 & DiD & 0.752 & 0.223 & 0.632 & 0.87 \\
1.00 & level & 0.195 & 0.058 & 0.612 & 0.35 \\
1.00 & DiD & 0.662 & 0.196 & 0.628 & 0.84 \\
\bottomrule
\end{tabular}

\vspace{3pt}
{\footnotesize Note: diagonal
$(\ell_{\mathrm{hy}},\ell_{\mathrm{ss}})=(0.61,0.24)$ and $\theta=0.88$
baseline from~\cite{zhao2026}; ``level'' basis uses the post-launch
cross-coefficient alone (conservative), ``DiD'' the pre--post change
(upper bound); reverse coupling scaled by the capital ratio $6.6/37.5$;
share = fraction of Samsung's closing displacement variance imported
from the Hynix channel via \eqref{eq:share}; shares use the assembled
full-matrix $\mathbf{M}$, so the small reverse coupling is handled
exactly rather than by one-way formulas. The capital-ratio scaling of the
reverse coupling assumes reciprocal cross-impact
($\phi_{12}=\phi_{21}$) and comparable liquidity normalization; it is
second order for every reported statistic. Rows shown are the
conservative $\mathbf{M}$-mapping; the $\mathbf{L}$-mapping of
qualification~(iv) raises the conservative-basis share to $0.84$--$0.91$
across $\theta$ and $\rho$ to at most $0.644$ ($0.86$ on the upper-bound
basis).}
\end{table}

Assembling $\hat{\mathbf{L}}$ with own-gains from the external
calibration, our cross estimate mapped through
Proposition~\ref{prop:ident}, and the reverse coupling scaled by the
capital ratio, two results follow (Table~\ref{tab:quant})---one of them
deliberately deflating. Before importing the diagonal we attempt to
estimate it ourselves: applying the same DiD logic to the \emph{own}
coefficients (pre-launch own-reversal as the generic microstructure
baseline) yields a self-contained Hynix gain of
$\hat\ell_{\mathrm{hy}}=0.42$ ($z=-2.0$), reassuringly close to the
externally calibrated $0.61$, but the Samsung own-effect is
undetectable at this sample size (wrong-signed, $|z|<0.5$), so a fully
self-assembled matrix is not yet usable and we retain the calibrated
diagonal.

\emph{The spectral-radius excess in Korea is small:} under the
conservative $\mathbf{M}$-mapping,
$\rho(\hat{\mathbf{L}})=0.612$--$0.644$ against a scalar maximum of
$0.61$, an excess of $+0.002$ to $+0.034$; under the $\mathbf{L}$-mapping
of qualification~(iv) the sensitivity range widens to at most $0.644$ on the
conservative basis ($0.86$ on the upper-bound basis), still bounded well
away from the static pole.
The reason is structural and follows from
Proposition~\ref{prop:rho}: detected coupling is strongly
one-directional, the matrix is quasi-triangular, and cycle amplification
requires a two-way product. Because the excess scales with the weakly
identified coupling product ($\ell_{12}\ell_{21}/|\ell_{11}-\ell_{22}|$
to leading order here), we do not attach a confidence
interval to it beyond the variant range shown; the honest summary is
that Korean cycle amplification is statistically indistinguishable from
zero, by the structure of the episode rather than by noise. The cycle
blind spot is therefore a warning
about \emph{symmetric} complexes, not a feature of the Korean
configuration---which motivates Section~\ref{sec:us}.

\emph{The transmission blind spot is large:} the headline
$M_{21}\approx 0.22$ means one unit of Hynix news transmits a 22\% pure
displacement into Samsung's close (Proposition~\ref{prop:transmit}). The
variance of the receiver's closing displacement decomposes across
channels as
\begin{equation}
\mathrm{share}_{2\leftarrow 1} \;=\;
\frac{M_{21}^{2}\,\sigma_{1}^{2}}
     {M_{21}^{2}\,\sigma_{1}^{2} + M_{22}^{2}\,\sigma_{2}^{2}},
\label{eq:share}
\end{equation}
with $\sigma_i$ the post-launch daily volatility of $r_{1,i}$ and
$M_{22}=\ell_{22}/(1-\ell_{22})$ the receiver's own overshoot. Because
the two inputs are correlated ($0.86$ in the pre-period), a covariance
cross-term of the same sign exists and channel attribution is not
unique; \eqref{eq:share} reports the standard component-variance
attribution, which allocates none of the shared movement to either
channel. Evaluated
at the data, \eqref{eq:share} gives ${\sim}41\%$ of
Samsung's closing displacement variance imported from the Hynix
channel on the conservative basis (30--84\% across the delta-method
interval on the post coefficient; ${\sim}87\%$ on the upper-bound basis).
Samsung's own scalar gain of $0.24$---a multiplier of $1.32$,
``moderate'' by any per-product standard---is silent about all of it.

\section{Evidence II: the U.S. Symmetric Complex}\label{sec:us}

The U.S.\ MicroStrategy--Bitcoin--Coinbase complexes are levered claims
on one risk factor and, unlike Korea, approximately symmetric---the
configuration for which Proposition~\ref{prop:rho} makes cycle
amplification largest. Using CRSP daily data (January 2023--December
2024) for seven funds (MSTU, MSTZ, MSTX with its leverage change from
$1.75\times$ to $2\times$ on October~29,~2024, BITX, BITU, BITI, CONL),
the underlyings (MSTR, COIN, and BITO as the exchange-hours Bitcoin
proxy), placebo pairs (XOM--CVX, JPM--GS), and QQQ overnight returns as
the window control, we re-run design \eqref{eq:reg} with the DiD split at
the MSTX launch (August~15,~2024); the post period contains the
November--December 2024 MSTR episode.

\begin{table}[t]
\caption{U.S. symmetric complex: capital and cross-reversal results.}
\vspace{2pt}
\label{tab:us}
\centering
\small
\setlength{\tabcolsep}{3.5pt}
\begin{tabular}{@{}lccc@{}}
\toprule
\multicolumn{4}{@{}l}{\emph{Rebalancing capital (scaled), post-launch}}\\
Complex & mean $\gamma$ & max $\gamma$ & 2024Q4 mean \\
\midrule
MSTR & 0.77 & \textbf{1.72} & 0.92 \\
COIN & 0.65 & 0.98 & 0.67 \\
\midrule
\multicolumn{4}{@{}l}{\emph{Cross-reversal DiD (treated directions)}}\\
Direction & pre $t$ & post $\hat\beta_\times$ ($t$) & DiD $z$ \\
\midrule
BITO$\to$MSTR & $+1.2$ & $-0.001\ (-0.0)$ & $-1.06$ \\
MSTR$\to$BITO & $-0.7$ & $-0.020\ (-0.4)$ & $-0.36$ \\
MSTR$\to$COIN & $+0.6$ & $+0.001\ (+0.0)$ & $-0.06$ \\
COIN$\to$MSTR & $+1.2$ & $-0.066\ (-0.9)$ & $-1.45$ \\
BITO$\to$COIN & $+1.7$ & $+0.040\ (+0.5)$ & $-0.54$ \\
COIN$\to$BITO & $+0.1$ & $+0.044\ (+0.6)$ & $+0.52$ \\
\midrule
Placebos (4 directions) & & & $|z|\le 1.22$ \\
\bottomrule
\end{tabular}

\vspace{3pt}
{\footnotesize Note: CRSP daily, 2023--2024; DiD split at the MSTX
launch (Aug.~15,~2024); $\gamma$ statistics are post-launch; the
BTC-complex $\gamma$ is omitted because those funds rebalance in CME
futures, for which an ETF-based denominator is invalid.}
\end{table}

Table~\ref{tab:us} reports the result, and Fig.~\ref{fig:forest}
collects every tested direction across both markets in one view: the
MSTR complex's scaled capital reaches $1.72$ in the mania---comparable to
Korean levels---yet all six treated cross-reversal directions are null,
as are placebos.

\begin{figure}[t]
\centering
\includegraphics[width=0.95\columnwidth]{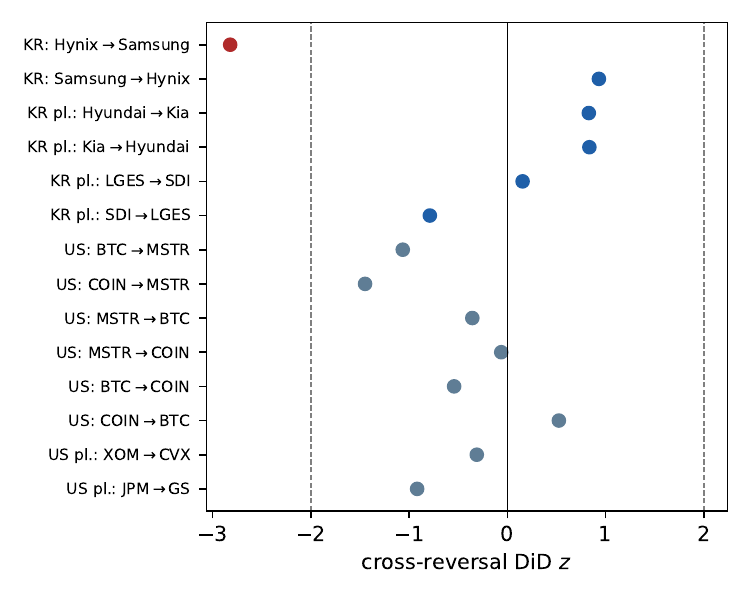}
\caption{All tested cross-reversal directions in one view: DiD $z$ per
ordered direction (overnight window, frozen specification), Korea and
U.S., treated and placebo. Dashed lines at $\pm 2$. Exactly one
direction exceeds the band---Hynix$\to$Samsung, the direction predicted
by relative complex size in the high-impact market.}
\label{fig:forest}
\end{figure} The null is
informative rather than disappointing. The loop gain is $\Lambda_c K$,
not $K$: U.S.\ closing auctions absorb flow at impact coefficients an
order of magnitude below Korea's (scalar gains 0.03--0.22 versus
0.61~\cite{zhao2026}), and by \eqref{eq:rho2} coupling terms inherit the
product of impacts. Coupling, like the scalar gain, is graded by venue
depth---and an estimator that ``found'' coupling wherever assets are
correlated would have failed here; the U.S.\ panel therefore doubles as a
scale placebo for the machinery. Three caveats are recorded: the U.S.\
overnight window contains continuous crypto trading (uncontrolled, which
biases toward zero and is conservative for our null); the pre-period
already contains the BITX complex, so the DiD tests the MSTR-launch
increment only; and BITO, our exchange-hours Bitcoin proxy, is a
futures-based fund whose roll yield adds noise to the sender return.
The sample ends with our CRSP coverage (December
2024), excluding the 2025--26 growth of these complexes.

\section{Monitoring Implications, Limitations, and
Conclusion}\label{sec:policy}

Regulatory responses to the Korean episode---capacity caps,
participation restrictions, flexible leverage---are being designed
product by product, i.e., in scalar terms, and the monitoring dashboards
that motivated them are organized the same way. Our results imply that
the object of surveillance should be the loop-gain matrix of each
(complex $\times$ clearing venue) system. Its spectral radius bounds
cycle amplification and, for nonnegative coupling, is guaranteed by
Proposition~\ref{prop:rho} to be
understated by any scalar audit, with the largest understatement for
symmetric complexes on a common factor---precisely the configuration now
growing in U.S.\ crypto-linked products. Its resolvent off-diagonals
price the displacement each asset imports from its neighbors' products
(Proposition~\ref{prop:transmit})---the channel that, in Korea, accounts
for nearly half of a ``moderate-gain'' stock's closing displacement
variance. Both objects are estimable in reduced form from prices and
public fund disclosures alone, which makes matrix monitoring
implementable by exchanges and supervisors in real time, without the
proprietary order-flow data that direct impact estimation requires.

Three limitations bound the claims, and each marks future work. First,
although the self-contained own-gain estimate agrees with the external
calibration where the data have power (Hynix, $0.42$ against $0.61$),
the Samsung diagonal and the correction share $\theta$ still lean on
external calibration~\cite{zhao2026}; joint estimation over a longer
panel is the natural next step. Second, the Korean post-launch sample is
short (56--60 trading days), which is why inference rests on exact
randomization, HAC and HC results are shown side by side, and every
attempted specification is reported. Third, the structural
interpretation of $\boldsymbol{\Phi}$ as venue cross-impact is
maintained rather than tested; distinguishing index-arbitrage,
inventory, and hedging channels requires flow attribution we
deliberately avoid. The dynamic extension of
Section~\ref{sec:dynamic} lowers the operative threshold to
\eqref{eq:dynthresh}, and the interaction of LETF complexes with other
self-referencing mandates---option gamma hedging, volatility targeting,
portfolio insurance~\cite{gennotte1990,shleifer2011}---forms exactly the
two-way cycles for which the spectral-radius blind spot is largest.

The one-sentence conclusion survives both the positive and the null
results: market stability is a property of the loop-gain matrix, and any
monitor that reads products one at a time is reading the diagonal of a
system whose risk lives off it.

\end{document}